\documentclass[10pt]{article}
\usepackage{amsmath,amssymb,amsthm,amsfonts}
\usepackage{cite}

\DeclareMathOperator{\res}{Res}
\DeclareMathOperator{\cas}{cas}
\DeclareMathOperator{\Sm}{S}
\DeclareMathOperator{\dr}{d}

\newcommand{\shft}[2]{#1^{{}^{(#2)}}}
\newcommand{\wvopinf}{\shft{\hat{W}}{\infty}}
\newcommand{\wvopzero}{\shft{\hat{W}}{0}}

\newcommand{\bakerinf}{\shft{w}{\infty}}
\newcommand{\bakerzero}{\shft{w}{0}}
\newcommand{\hbakinf}{\shft{\hat{w}}{\infty}}
\newcommand{\hbakzero}{\shft{\hat{w}}{0}}
\newcommand{\invDtp}{\Dt^{-1}_+}
\newcommand{\invDtm}{\Dt^{-1}_-}
\newcommand{\bx}{\bar{x}}
\newcommand{\by}{\bar{y}}
\newcommand{\bB}{\bar{B}}
\newcommand{\bC}{\bar{C}}

\newcommand{\tld}{\tilde}
\newcommand{\tcas}{\widetilde{\cas}}
\newcommand{\bcas}{\overline{\cas}}

\newcommand{\mc}{\mathcal}
\newcommand{\mb}{\mathbb}

\newcommand{\om}{\omega}
\newcommand{\Om}{\Omega}
\newcommand{\sg}{\sigma}
\newcommand{\ta}{\tau}

\newcommand{\p}{\partial}
\newcommand{\ld}{\lambda}
\newcommand{\Ld}{\Lambda}

\newcommand{\af}{\alpha}

\newcommand{\Dt}{\Delta}

\newcommand{\suml}{\sum\limits}

\newtheorem{theorem}{Theorem}
\newtheorem{proposition}{Proposition}
\newtheorem{definition}{Definition}

\newtheorem{lemma}{Lemma}
\newtheorem{example}{Example}

\title{A new multi-component two dimensional Toda lattice hierarchy and two dimensional Toda lattice with
  self-consistent sources}
\author{Xiaojun Liu\footnote{lxj98@mails.tsinghua.edu.cn}\\
  Department of Applied Mathematics,\\
  China Agricultural University, Beijing, 100083, PRC \and
  Yunbo Zeng\footnote{yzeng@math.tsinghua.edu.cn} \; and Runliang Lin\footnote{rlin@math.tsinghua.edu.cn}\\
  Department of Mathematical Sciences,\\
  Tsinghua University, Beijing, 100084, PRC}

\begin{document}

\maketitle

\begin{abstract}
  We propose a new multi-component two-dimensional Toda lattice hierarchy (mc2dTLH) which includes
  two-dimensional Toda lattice equation with self-consistent sources (2dTLSCS) as the first non-trivial
  equation. The Lax representations for this mc2dTLH are presented. We
  construct a non-auto-B\"acklund Darboux transformation (DT) for 2dTLSCS by applying the method of variation
  of constant (MVC) to ordinary DT of 2dTLSCS. This non-auto-B\"acklund DT enables us to obtain various
  solutions such as solitons, rational solutions etc., to 2dTLSCS.
\end{abstract}

\section{Introduction}
Multi-component generalizations of soliton equations attract a lot of attention from both physical and
mathematical points of view
\cite{MR638807,MR723457,MR730247,MR688946,MR2006751,MR1621464,MR1629543,MR810623}. The multi-component KP
(mcKP) hierarchy given by \cite{MR638807} contains physically relevent nonlinear integrale systems such as
Davey-Stewartson equation, two-dimensional Toda lattice and three-wave resonant interaction equation. The
multi-component Toda lattice Hierarchy \cite{MR810623} contains non-abelian Toda lattice equation. There exist
several equivalent formulations of this multi-component soliton equations. For example, there are matrix
pseudo-differential operator (Sato) formulation, $\ta$-function approach via matrix Hirota bilinear
identities, multi-component free fermion formulation for mcKP hierarchy. For two dimensional Toda lattice
hierarchy (2dTLH), a similar matrix-difference operator approach to multi-component hierarchy was also
presented by \cite{MR810623}.

Another kind of multi-component generalizations to soliton equations
are the so-called soliton equation with self-consistent sources
(SESCS), which were initiated by V.K. Mel'nikov
\cite{MR708435,MR910584,MR940618}. For two dimensional Toda lattice
equation (2dTL), the corresponding 2dTL with self-consistent sources
(2dTLSCS) was first presented in \cite{MR2261273} by source
generating method as follows
\begin{subequations}
  \label{eqns:2dTLSCS}
  \begin{align}
    q_{xy}&=e^{q-\shft{q}{-1}}-e^{\shft{q}{1}-q}+\sum(w_i w_i^*)_y,\\
    w_{i,y} &=e^{q-\shft{q}{-1}} \shft{w_i}{-1}\quad (i=1,\ldots,N),\\
    w^*_{i,y}&=-e^{\shft{q}{1}-q}\shft{{w_i^*}}{1}.
  \end{align}
\end{subequations}

In \cite{nmcKP}, we proposed a method to construct a new multi-component KP hierarchy which includes two kinds
of KP equation with self-consistent sources presented by Mel'nikov \cite{MR708435}. In this paper, we will
present a new multi-component 2dTLH which includes 2dTLSCS as the first non-trivial equation.

We briefly recall the framework of \cite{MR810623} as follows. Let ${\bf x}=(x_1,x_2,\cdots)$ and ${\bf
  y}=(y_1,y_2,\cdots)$ be two series of variables. Let $n\in\mb{Z}$ be a discrete variable. Then Lax equation
of 2dTLH is given by
\begin{subequations}
  \label{eqns:2TodaHierInf}
  \begin{align}
    L_{x_m} & = [B_m, L],\\
    L_{y_m} & = [C_m, L],\\
    M_{x_m} & = [B_m, M],\\
    M_{y_m} & = [C_m, M],
  \end{align}
\end{subequations}
where
\begin{align*}
  L &= \Ld + u_0 + u_1 \Ld^{-1} + u_2 \Ld^{-2} + \cdots,\\
  M &= v_{-1}\Ld^{-1} + v_0 + v_1\Ld + v_2\Ld^2 + \cdots,
\end{align*}
$\Ld$ is a \emph{shift operator} such that $\Ld f(n)=f(n+1)\Ld$, $u_i$ and $v_i$ are functions of ${\bf x}$,
${\bf y}$ and $n$, $B_m=L^m_+$ stands for the positive part ($\ge0$) of $L^m$ with respect to the powers of
$\Ld$ and $C_m=M^m_-$ stands for the negative part ($<0$) of $M^m$. The commutativety of
(\ref{eqns:2TodaHierInf}) gives rise to zero-curvature equations of 2dTLH
\begin{subequations}
  \label{eqns:2TodaZerCur}
  \begin{align}
    B_{k,x_m} - B_{m,x_k} + [B_k, B_m ] = 0 , \label{eqn:2TodaBB}\\
    C_{k,y_m} - C_{m,y_k} + [C_k, C_m ] = 0 , \label{eqn:2TodaCC}\\
    B_{k,y_m} - C_{m,x_k} + [B_k, C_m ] = 0 . \label{eqn:2TodaBC}
  \end{align}
\end{subequations}
When $m=k=1$, (\ref{eqns:2TodaZerCur}) leads to following 2-dimensional Toda equation
\begin{subequations}
  \label{eqns:2TodaEqns_1stOrder}
  \begin{align}
    u_y &= v - \shft{v}{1},\\
    v_x &= v\left(u - \shft{u}{-1}\right),\label{eqn:2TodaEqns_1stOrder-uvrel}
  \end{align}
\end{subequations}
and Lax pair for \eqref{eqns:2TodaEqns_1stOrder} is
\begin{subequations}
  \label{eqns:2TodaLax}
  \begin{align}
    \psi_x&=B(\psi)=(\Ld+u)(\psi),\\
    \psi_y&=C(\psi)=(v\Ld^{-1})(\psi),
  \end{align}
\end{subequations}
where $x:=x_1$, $y:=y_1$, $B:=B_1$, $C:=C_1$, $u:=u_0$, $v:=v_{-1}$. Eliminating $u$ from
(\ref{eqns:2TodaEqns_1stOrder}) and introducing $q:=q(n,x,y)$ which satisfies
\begin{equation}
  \label{eqn:v-q}
  v := \exp \left( q- \shft{q}{-1}\right),
\end{equation}
then \eqref{eqns:2TodaEqns_1stOrder} gives the so-called \emph{two dimensional Toda lattice equation}:
\begin{equation}
  \label{eqn:2TodaEqn_2stOrder}
  q_{xy} = \exp\left(q-\shft{q}{-1}\right)-\exp\left(\shft{q}{1} - q\right).
\end{equation}

Our multi-component generalization to 2dTLH can be presented as follows. We first introduce a new vector field
$\p_{\by_k}$, which is a linear combination of all vector fields $\p_{y_m}$. Then we get a new Lax type
equation which consists of $\p_{\by_k}$-flow and evolutions of wave functions. Under the evolutions of wave
functions, the commutativety of $\p_{\by_k}$, $\p_{y_m}$ and $\p_{x_k}$ flow give rise to the new
multi-component 2dTL hierarchy (mc2dTLH). This hierarchy enables us to derive the 2dTLSCS in different way
from \cite{MR2261273,2-toda-scs-casorati} and to obtain their Lax representations. This hierarchy is also
different from mc2dTLH given by \cite{MR810623}.

In the second part of our paper, we solve 2dTLSCS by means of
Darboux transformations (DT). Since the Lax representation for
2dTLSCS is obtained, we can construct an auto-B\"acklund DT for
2dTLSCS, which transforms between the solution of 2dTLSCS with same
number of source terms. However, such auto-B\"acklund transformation
can not be used to construct non-trivial solution from the trivial
solution.  The idea for us to construct non-auto-B\"acklund DTs is
to consider 2dTLSCS as 2dTL with non-homogeneous terms (i.e.
self-consistent source terms). Inspired by ODE method, we can apply
the method of variation of constant (MVC) to DT to find a new
non-auto-B\"acklund DT which transforms original solution of 2dTLSCS
with $N$ self-consistent sources to a new solution of 2dTLSCS with
$(N+1)$ self-consistent sources. This non-auto-B\"acklund DT enables
us to find various solutions for 2dTLSCS. Furthermore, we obtained
the $m$-time repeated non-auto-B\"acklund DTs formula, and exhibit
some solutions of 2dTLSCS which include solitons, rational solutions
and etc.

Our paper will organized as follows. In section \ref{sec:mc2dTLH} we propose resolvent identities and present
new mc2dTLH which includes 2dTLSCS. In section \ref{sec:DT-2dTLSCS} we first construct auto-B\"acklund DT for
2dTLSCS. Then by applying method of variation of constant to this DT, we find a non-auto-B\"acklund DT which
can increase the number of source term by $1$. We obtain $m$-time repeated non-auto-B\"acklund DT formula
which can be expressed in compact form with Casoratian determinants. In section \ref{sec:solutions}, we
present some solutions to 2dTLSCS by using this $m$-time repeated DTs.

\section{New Multi-component 2-dimensional Toda Lattice Hierarchy}
\label{sec:mc2dTLH}

\subsection{Sato approach and resolvent identities}

First we introduce some useful notations and definitions which can be found in \cite{MR810623}.
\begin{definition}[The residue of shift operator]
  Let $P=\sum_{i\in\mb{Z}}P_i\Ld^i$, then residue of $P$ is
  \begin{displaymath}
    \res_\Ld P = P_0.
  \end{displaymath}
\end{definition}
\begin{definition}[The adjoint operator ${}^*$]
  $P^* = \sum_{i\in\mb{Z}} \Ld^{-i}P_i$.
\end{definition}
\begin{definition}[Shift operator's action]
  The shift operator action $P(\ld^n)$ can be defined as
  \begin{displaymath}
    P(\ld^n) = \sum_{i\in\mb{Z}}
    P_i\Ld^i(\ld^n)=(\sum_{i\in\mb{Z}}P_i\ld^i)\cdot \ld^n.
  \end{displaymath}
\end{definition}
\begin{definition}[Formal inversion of difference operator $\Dt$]
  For difference operator $\Dt=\Ld -1$ , the formal inversion are given by $$\Dt_+^{-1}=-\sum_{i\ge0}\Ld^i
  \quad \text{or}\quad\Dt_-^{-1}=\sum_{i\le-1}\Ld^i.$$
\end{definition}
Introduce \emph{wave operators}
\begin{align*}
  \wvopinf &= b_0 + b_1\Ld^{-1} + b_2\Ld^{-2} + \cdots,\\
  \wvopzero & = c_0 + c_1\Ld + c_2\Ld^2 + \cdots,
\end{align*}
where $b_0=1$. Ueno and Takasaki proved \cite{MR810623} that if $L$ and $M$ are solutions to
(\ref{eqns:2TodaHierInf}) then there exist wave operators, such that $L$ and $M$ can be written as
\begin{subequations}
  \begin{equation}
    \label{eqn:Sandwich}
    L = \wvopinf \Ld \wvopinf{}^{-1},\quad
    M = \wvopzero \Ld^{-1} \wvopzero{}^{-1},
  \end{equation}
  and
  \begin{alignat}{2}
    \p_{x_m}\wvopinf &=-L^m_{<0} \wvopinf,&\quad
    \p_{x_m}\wvopzero &=L^m_{\ge 0}\wvopzero,\\
    \p_{y_m}\wvopinf &= M^m_{<0}\wvopinf, &\quad \p_{y_m}\wvopzero &= -M^m_{\ge0} \wvopzero.
  \end{alignat}
\end{subequations}
Define \emph{wave function}
\begin{align*}
  \bakerinf &= \wvopinf(\ld^n)e^{\xi(x,\ld)}=(\sum_{i\ge0}b_i\ld^{-i})\ld^ne^{\xi(x,\ld)}
  := \hbakinf \ld^n e^{\xi(x,\ld)},\\
  \bakerzero &= \wvopzero(\ld^n)e^{\xi(y,\ld^{-1})}=(\sum_{i\ge0}c_i\ld^i)\ld^ne^{\xi(y,\ld^{-1})}
  := \hbakzero \ld^n e^{\xi(y,\ld^{-1})},
\end{align*}
and \emph{adjoint wave function}\\
\begin{align*}
  \bakerinf{}^* &=
  \wvopinf{{}^*}^{-1}(\ld^{-n})e^{-\xi(x,\ld)}:=\hbakinf{}^*\ld^{-n}e^{-\xi(x,\ld)},\\
  \bakerzero{}^* &= \wvopzero {{}^*}^{-1}(\ld^{-n}) e^{-\xi(y,\ld^{-1})} := \hbakzero{}^*\ld^{-n}
  e^{-\xi(y,\ld^{-1})},
\end{align*}
where $\xi(x,\ld)=\sum_{i\ge1}x_i\ld^i$,
$\xi(y,\ld^{-1})=\sum_{i\ge1}y_i\ld^{-i}$.  Then
(\ref{eqns:2TodaHierInf}) can also be given as the compatibility
condition of the following linear evolution equations
\begin{subequations}
  \begin{alignat}{4}
    L \bakerinf &= \ld \bakerinf,
    \quad\quad&M\bakerzero&=\ld^{-1}\bakerzero,\\
    \p_{x_m}\bakerinf &= B_m\bakerinf,
    \quad\quad&\p_{x_m}\bakerzero&= B_m\bakerzero,\\
    \p_{y_m}\bakerinf &= C_m\bakerinf, \quad\quad&\p_{y_m}\bakerzero&= C_m\bakerzero.
  \end{alignat}
\end{subequations}
\begin{lemma}[Ueno and Takasaki\cite{MR810623}]
  Suppose $P=\sum P_i\Ld^i$, $Q=\sum Q_j\Ld^j$, then
  \begin{displaymath}
    \res_\Ld P\cdot Q^* = \res_\ld\ld^{-1}P(\ld^n)\cdot Q(\ld^{-n}).
  \end{displaymath}
\end{lemma}
\begin{proof}
  Only need to show for $P=P_i\Ld^i$, $Q=Q_j\Ld^j$. $\res_\Ld P Q^*=\res_\Ld P_i\Ld^{i-j}
  Q_j=\delta_{i,j}P_iQ_j$, while $\res_\ld \ld^{-1}P(\ld^n)Q(\ld^{-n})
  =\res_\ld\ld^{-1}P_i\ld^{n+i}Q_j\ld^{-n-j}=\delta_{i,j}P_iQ_j$.
\end{proof}
Similar to the KP theory, in which the principle part of resolvent can be expressed in terms of a quadratic
form of wave function and adjoint wave function \cite{MR1964513}, we have the following resolvent identities
for 2dTLH.
\begin{proposition}[Resolvent identities]
  \begin{alignat*}{4}
    \sum_{k\ge0}L^k_{\ge0}\ld^{-k} &=-\bakerinf\invDtp\bakerinf{}^*,\quad\quad
    &\sum_{k\in\mb{Z}}M^k_{\ge0}\ld^k&=-\bakerzero\invDtp
    \bakerzero{}^*,\\
    \sum_{k\in\mb{Z}}L^k_{<0}\ld^{-k} &=\bakerinf\invDtm\bakerinf{}^*,\quad\quad
    &\sum_{k>0}M^k_{<0}\ld^k&=\bakerzero\invDtm \bakerzero{}^*.
  \end{alignat*}
\end{proposition}
\begin{proof}
  We prove one of them, others are similar. Since $L=\wvopinf\Ld\wvopinf{}^{-1}$, {\allowdisplaybreaks
    \begin{align*}
      L^k_{<0} &=\left(\wvopinf\Ld^k\wvopinf{}^{-1}\right)_{<0}\\
      &=\sum_{m\ge1}\res_\Ld
      \left(\wvopinf\Ld^k\wvopinf{}^{-1}\Ld^m\right)\Ld^{-m}\\
      &=\sum_{m\ge1}\res_\ld \ld^{-1}\left(\wvopinf\Ld^k(\ld^n)e^{\xi(x,\ld)}\right)
      \cdot\left(\Ld^{-m}\wvopinf{{}^{*}}^{-1}(\ld^{-n})e^{-\xi(x,\ld)}\right)
      \Ld^{-m}\\
      &=\sum_{m\ge1}\res_\ld \ld^{k-1}\bakerinf\Ld^{-m}\bakerinf{}^*\\
      &=\res_\ld\ld^{k-1}\bakerinf\invDtm\bakerinf{}^*.
    \end{align*}} So $\sum_{k\in\mb{Z}}L^k_{<0}\ld^{-k}=\bakerinf\invDtm\bakerinf{}^*$.
\end{proof}

\subsection{New mc2dTL hierarchy}

For fixed $k\ge1$, $N>0$, we define a new time variable $\by_k$ such that the corresponding vector field is
\begin{equation}
  \label{eqn:newy}
  \p_{\by_k}=\p_{y_k}+\sum_{i=1}^N\sum_{j>1}\ld_i^j\p_{y_j}
\end{equation}
where $\ld_i$ are distinct arbitrary non-zero parameters. Then the $\by_k$ flow is given by
\begin{equation}
  \label{eqns:new-evolution}
  \frac{\p}{\p_{\by_k}}L=[\bC_k,L],\quad\quad \frac{\p}{\p_{\by_k}}M=[\bC_k,M],
\end{equation}
with
\begin{displaymath}
  \bC_k=C_k+\sum_{i=1}^N\sum_{j\ge1}\ld_i^jC_j,
\end{displaymath}
which, according to Proposition 1, can be rewritten as
\begin{displaymath}
  \bC_k=C_k+\suml_{i=1}^N\bakerzero_i\invDtm {\bakerzero_i}^*.
\end{displaymath}
By setting $w_i=\bakerzero_i$, $w_i^*={\bakerzero_i}^*$, the compatibilities of (\ref{eqns:2TodaHierInf}) and
(\ref{eqns:new-evolution}) give rise to the following {\it new multi-component two dimensional Toda lattice
  hierarchy}.
\begin{proposition} We have the following new mc2dTLH,\\
  for $m\neq k$:
  \begin{subequations}
    \label{eqns:mc2dTLH-m-neq-k}
    \begin{align}
      &B_{m,x_k}-B_{k,x_m}+[B_m,B_k]=0\label{eqn:mc2dTLH-m-neq-k-1},\\
      &C_{m,\by_k}-\bC_{k,y_m}+[C_m,\bC_k]=0\label{eqn:mc2dTLH-m-neq-k-2},\\
      &B_{m,\by_k}-\bC_{k,x_m}+[B_m,\bC_k]=0\label{eqn:mc2dTLH-m-neq-k-3},\\
      &B_{k,y_m}-C_{m,x_k}+[B_k,C_m]=0\label{eqn:mc2dTLH-m-neq-k-4},\\
      &w_{i,x_m}=B_m(w_i)\quad\quad w_{i,y_m}=C_m(w_i) \quad
      (i=1,\ldots,N)\label{eqn:mc2dTLH-m-neq-k-wv},\\
      &w^*_{i,x_m}=-B_m^*(w_i^*),\quad\quad w^*_{i,y_m}=-C_m^*(w_i^*),\label{eqn:mc2dTLH-m-neq-k-adjwv}
    \end{align}
  \end{subequations}
  for $m=k$:
  \begin{subequations}
    \label{eqns:mc2dTLH-m-eq-k}
    \begin{align}
      &B_{k,\by_k}-\bC_{k,x_k}+[B_k,\bC_k]=0\label{eqn:mc2dTLH-m-eq-k-1}\\
      &\p_{x_k}w_i=B_k(w_i),\quad\p_{x_k}w_i^*=-B_k^*(w_i^*)\quad(i=1,\ldots,N)\label{eqn:mc2dTLH-m-eq-k-wv}
    \end{align}
  \end{subequations}
  where $\bC_k=C_k+\suml_{i=1}^Nw_i\invDtm w_i^*$.
\end{proposition}
It is worth noting that $w_i$ and $w_i^*$ need not necessarily to be the wave function and adjoint wave
function. In fact, the equations (\ref{eqn:mc2dTLH-m-neq-k-wv}) and (\ref{eqn:mc2dTLH-m-neq-k-adjwv}) ( or
(\ref{eqn:mc2dTLH-m-eq-k-wv})) ensure the closeness of
(\ref{eqn:mc2dTLH-m-neq-k-1})-(\ref{eqn:mc2dTLH-m-neq-k-4}) (or (\ref{eqn:mc2dTLH-m-eq-k-1})). Furthermore,
under the conditions (\ref{eqn:mc2dTLH-m-neq-k-wv}) and (\ref{eqn:mc2dTLH-m-neq-k-adjwv}) (or
(\ref{eqn:mc2dTLH-m-eq-k-wv})), one can easily obtains the Lax representations of
(\ref{eqn:mc2dTLH-m-neq-k-1})-(\ref{eqn:mc2dTLH-m-neq-k-4}) as
\begin{subequations}
  \begin{alignat}{4}
    \psi_{x_m}&=B_m(\psi),\quad\quad \psi_{x_k}&=&B_k(\psi),\\
    \psi_{y_m}&=C_m(\psi),\quad\quad \psi_{\by_k}&=&\bC_k(\psi),
  \end{alignat}
\end{subequations}
or get the Lax representation of (\ref{eqn:mc2dTLH-m-eq-k-1}) as
\begin{equation}
  \psi_{x_k}=B_k(\psi),\quad\psi_{\by_k}=\bC_k(\psi).
\end{equation}

\begin{example}[Two dimensional Toda lattice equation with Self-consistent sources]
  When $m=k=1$ , let $u=u_0$, $v=v_{-1}$, $x=x_1$, $y=\by_1$
  \begin{displaymath}
    B_1=\Ld+u,\quad C_1=v\Ld^{-1}.
  \end{displaymath}
  then (\ref{eqns:mc2dTLH-m-eq-k}) becomes
  \begin{subequations}
    \label{eqns:2dTLSCS-1}
    \begin{align}
      &u_y=-\Dt(v+\sum_{i=1}^N w_iw_i^*\shft{{}}{-1}),\quad v_x=v(u-\shft{u}{-1}),\\
      &w_{i,x}=B_1(w_i),\quad w^*_{i,x}=-B_1^*(w_i^*), \quad i=1,\ldots,N
    \end{align}
  \end{subequations}
  Under $u=q_x$, $v=\exp(q-\shft{q}{-1})$, (\ref{eqns:2dTLSCS-1}) yields
  \begin{subequations}
    \label{eqns:2dTLSCS-q}
    \begin{align}
      q_{xy}&=e^{q-\shft{q}{-1}}-e^{\shft{q}{1}-q}+\sum_{i=1}^N(w_i w_i^*)_x,\\
      w_{i,x} &=\shft{w_i}{1}+q_xw_i,\quad (i=1,\ldots,N)\\
      w^*_{i,x}&=-w_i^*\shft{{}}{-1}-q_xw_i^*.
    \end{align}
  \end{subequations}
  This is two dimensional Toda lattice equation with $N$ self-consistent sources (2dTLSCS).
\end{example}

Analogously, let us introduce $\bx_k$, such that
\begin{displaymath}
  \p_{\bx_k}= \p_{x_k}+\sum_{i=1}^N\sum_{j\ge1}\ld_i^{-j}\p_{x_j},
\end{displaymath}
then we will get another new multi-component two dimensional Toda lattice hierarchy.
\begin{proposition} We have another new mc2dTLH as follows,\\
  for $m\neq k$
  \begin{subequations}
    \begin{align}
      &B_{m,\bx_k}-\bB_{k,x_m}+[B_m,\bB_k]=0,\\
      &C_{m,\bx_k}-\bB_{k,y_m}+[C_m,\bB_k]=0,\\
      &C_{m,y_k}-C_{k,y_m}+[C_m,C_k]=0,\\
      &B_{m,y_k}-C_{k,x_m}+[B_m,C_k]=0,\\
      &w_{i,y_m}=C_m(w_i),\quad w_{i,x_m}=B_m(w_i), \quad  i=1,\ldots,N,\\
      &w^*_{i,y_m}=-C_m^*(w_i^*),\quad w^*_{i,x_m}=-B_m^*(w_i^*).
    \end{align}
  \end{subequations}
  for $m=k$
  \begin{subequations}
    \label{eqns:Another-mc2dTLH-m-eq-k}
    \begin{align}
      &C_{k,\bx_k}-\bB_{k,y_k}+[C_k,\bB_k]=0,\\
      &\p_{y_k}w_i=C_k(w_i),\quad\p_{y_k}w_i^*=-C_k^*(w_i^*).\quad i=1,\ldots,N,
    \end{align}
  \end{subequations}
  where
  \begin{displaymath}
    \bB_k=B_k-\suml_{i=1}^Nw_i\invDtp w_i^*.
  \end{displaymath}
\end{proposition}

\begin{example}[2dTLSCS \cite{MR2261273,2-toda-scs-casorati}]
  When $m=k=1$, (\ref{eqns:Another-mc2dTLH-m-eq-k}) leads to (\ref{eqns:2dTLSCS}). It is interesting to see
  that (\ref{eqns:2dTLSCS}) is equivalent to (\ref{eqns:2dTLSCS-q}) under
  \begin{align*}
    x&\to -y,\quad y\to -x,\quad q\to q,\\
    w_i&\to -e^qw_i^*,\quad w_i^*\to e^{-q}w_i.
  \end{align*}
  So hereafter we may concentrate on 2dTLSCS (\ref{eqns:2dTLSCS-1}). This transformation was discovered by
  Prof. Hu Xingbiao.
\end{example}

\section{Darboux transformation for 2dTLSCS}
\label{sec:DT-2dTLSCS}

In the second part of our paper, we concentrate on 2dTLSCS~(\ref{eqns:2dTLSCS-1}). First recall the Lax pair
of 2dTL equation (\ref{eqns:2TodaLax}). For convenience, hereafter we denote~$B=B_1$, $C=C_1$, $u=u_0$,
$v=v_{-1}$, $x=x_1$, $y=y_1$.

\subsection{Applying the method of variation of constant to DT of 2dTLSCS }

Let us first introduce the notion of Casoratian determinant: for $m$ discrete variable functions
$h_1,\cdots,h_m$, the Casoratian determinant
\begin{displaymath}
  \cas(h_1,\cdots,h_m)=
  \begin{vmatrix}
    h_1 &\cdots & h_m\\
    \shft{h_1}{1} & \cdots & \shft{h_m}{1}\\
    \vdots & \vdots & \vdots\\
    \shft{h_1}{m-1} & \cdots & \shft{h_m}{m-1}
  \end{vmatrix}.
\end{displaymath}

Darboux transformation for 2dTL (\ref{eqns:2TodaEqns_1stOrder}) was given in \cite{bk:MatveevSalle91}. Let us
first recall this DT and its proof as following Lemma.
\begin{lemma}
  \label{thm:DT-toda}
  Let $h$ be special solution to (\ref{eqns:2TodaLax}). Let $\mc{D}=\Ld+\sg$,$\sg:=-\shft{h}{1}/h$, then DT
  \begin{subequations}
    \label{eqns:DT-toda}
    \begin{align}
      \tld{u}&:=\shft{u}{1}+\sg-\shft{\sg}{1},\label{eqn:DT-Toda-u}\\
      \tld{v}&:=v\sg/\shft{\sg}{-1},\label{eqn:DT-Toda-v}\\
      \tld{\psi}&:=\mc{D}(\psi)=\frac{\cas(h,\psi)}{h},\label{eqn:DT-Toda-psi}
    \end{align}
  \end{subequations}
  gives a new solution to (\ref{eqns:2TodaLax}). Thus $\tld{u}$, $\tld{v}$ are new solution for
  (\ref{eqns:2TodaEqns_1stOrder}).
\end{lemma}
\begin{proof}
  Since $\tld{B}:=\Ld+\tld{u}$, $\tld{C}:=\tld{v}\Ld^{-1}$, $\tld{\psi}=\mc{D}(\psi)$, a sufficient condition
  such that (\ref{eqns:2TodaLax}) holds is
  \begin{subequations}
    \label{eqns:DT-cond}
    \begin{align}
      &\mc{D}_x+\mc{D}B-\tld{B}\mc{D}=0,\label{eqn:DT-cond-xpart}\\
      &\mc{D}_y+\mc{D}C-\tld{C}\mc{D}=0.
    \end{align}
  \end{subequations}
  Notice that
  \begin{equation}
    \label{eqn:annihilate}
    \mc{D}(h)=0,
  \end{equation}
  take partial derivative $\p_x$,$\p_y$ to (\ref{eqn:annihilate}), one gets
  \begin{align*}
    &\mc{D}_x(h)+\mc{D}(h_x)=(\mc{D}_x+\mc{D}B)(h)=\tld{B}\mc{D}(h)=0,\\
    &\mc{D}_y(h)+\mc{D}(h_y)=(\mc{D}_y+\mc{D}C)(h)=\tld{C}\mc{D}(h)=0,
  \end{align*}
  which mean
  \begin{subequations}
    \label{eqns:DT-equiv-cond}
    \begin{align}
      &(\mc{D}_x+\mc{D}B-\tld{B}\mc{D})(h)=0,\\
      &(\mc{D}_y+\mc{D}C-\tld{C}\mc{D})(h)=0.
    \end{align}
  \end{subequations}
  From (\ref{eqn:DT-Toda-u}) and (\ref{eqn:DT-Toda-v}) one knows the operators acting on $h$ in
  (\ref{eqns:DT-equiv-cond}) are scalar functions multiplications. So (\ref{eqns:DT-cond}) holds.
\end{proof}

The Lax representation for 2dTLSCS (\ref{eqns:2dTLSCS-1}) is
\begin{subequations}
  \label{eqns:Lax-2TodaEqScs}
  \begin{align}
    \psi_x&=B(\psi),\label{eqn:Lax-2TodaEqScs-xpart}\\
    \psi_y&=(C+\sum_{i=1}^Nw_i\invDtm w_i^*)(\psi).\label{eqn:Lax-2TodaEqScs-ypart}
  \end{align}
  Note that Lax representation (\ref{eqns:Lax-2TodaEqScs}) holds under following equations
  \begin{align}
    w_{i,x}&=B(w_i),\quad i=1,\ldots,N,\label{eqn:Lax-2TodaEqScs-src}\\
    w^*_{i,x}&= -B^*(w_i).\label{eqn:Lax-2TodaEqScs-adjsrc}
  \end{align}
\end{subequations}
\begin{proposition}[Darboux transformation for 2dTLSCS (\ref{eqns:2dTLSCS-1})]
  \label{thm:DT-src-adjsrc}
  Let $h$ be a special solution to (\ref{eqns:Lax-2TodaEqScs}), $\mc{D}=\Ld+\sg$, $\sg:=-\shft{h}{1}/h$. Based
  on the Darboux transformation (\ref{eqns:DT-toda}), define
  \begin{subequations}
    \label{eqns:DT-src-adjsrc}
    \begin{align}
      \tld{w}_i&:=\mc{D}(w_i)=\frac{\cas(h,w_i)}{h},\label{eqn:DT-src}\\
      \tld{w}_i^*&:={\mc{D}^{*}}^{-1} (w_i^*)=-\frac{\Sm(h w_i^*)}{\shft{h}{1}},\label{eqn:DT-adj-src}
    \end{align}
  \end{subequations}
  where $\Sm:=\Ld\invDtm$. Then (\ref{eqns:DT-toda}) and (\ref{eqns:DT-src-adjsrc}) together give a new
  solution to (\ref{eqns:Lax-2TodaEqScs}). Thus one gets a new solution to (\ref{eqns:2dTLSCS-1}).
\end{proposition}
\begin{proof}
  From Lemma \ref{thm:DT-toda} it is easy to see that $\tld{w}_i$ defined by (\ref{eqn:DT-src}) satisfies
  (\ref{eqn:Lax-2TodaEqScs-src}). It is necessary to prove $\tld{w}_i^*$ satisfies
  (\ref{eqn:Lax-2TodaEqScs-adjsrc}). From the proof of Lemma \ref{thm:DT-toda} we know
  \begin{displaymath}
    (\p_x-\tld{B})\mc{D}=\mc{D}(\p_x-B).
  \end{displaymath}
  Taking formal adjoint ${}^*$ to this equality and rewrite it as
  \begin{displaymath}
    (-\p_x-\tld{B}^*){\mc{D}^*}^{-1}={\mc{D}^*}^{-1}(-\p_x-B^*).
  \end{displaymath}
  This is a sufficient condition for ${\mc{D}^*}^{-1}$ to be the Darboux transformation for
  (\ref{eqn:Lax-2TodaEqScs-adjsrc}). Thus we have proved (\ref{eqn:DT-adj-src}). At last we need to prove that
  Darboux transformation given by (\ref{eqns:DT-toda}) and (\ref{eqns:DT-src-adjsrc}) fulfills
  (\ref{eqn:Lax-2TodaEqScs-ypart}). That is
  \begin{equation}
    \label{eqn:DT-cond-2TodaEqScs-main}
    \mc{D}_y+\mc{D}C+\sum_{i=1}^N\mc{D} w_i\invDtm w_i^*
    -\tld{C}\mc{D}-\sum_{i=1}^N\tld{w}_i\invDtm\tld{w}_i^*\mc{D}=0.
  \end{equation}
  Based on (\ref{eqn:DT-cond-xpart}), we have to prove the extra terms w.r.t. $w_i$, $w_i^*$ in
  (\ref{eqn:DT-cond-2TodaEqScs-main}) are equal. For every $i$, we have {\allowdisplaybreaks
  \begin{align*}
    &-\frac{\shft{w_i}{1}\Sm(hw_i^*)}{h}+\frac{\shft{h}{1}}{h^2}w_i\invDtm(hw_i^*)
    +\shft{w_i}{1}\Ld\invDtm w_i^*\\
    &-\frac{\shft{h}{1}}{h}w_i\invDtm
    w_i^*-\tld{w}_i\invDtm\tld{w}_i^*\Ld+\tld{w}_i\invDtm\tld{w}_i^*\frac{\shft{h}{1}}{h}\\
    =&-\frac{\tld{w}_i}{h}\Sm(h
    w_i^*)-\frac{\shft{h}{1}}{h}w_iw_i^*+\shft{w_i}{1}\Ld\invDtm
    w_i^*\\
    &-\frac{\shft{h}{1}}{h}w_i\invDtm
    w_i^*+\tld{w}_i\invDtm\frac{\Sm(h w_i^*)}{\shft{h}{1}}\Ld-
    \tld{w}_i\invDtm\frac{\Sm(h w_i^*)}{h}\\
    =&-\frac{\tld{w}_i}{h}\Sm(h w_i^*)+\shft{w_i}{1}\Ld\invDtm
    w_i^*-\frac{\shft{h}{1}}{h}w_i\Ld\invDtm w_i^*\\
    &+\tld{w}_i\invDtm\Ld\frac{\Sm(h
      w_i^*)}{h}-\tld{w}_i\invDtm\frac{\Sm(hw_i^*)}{h}-\tld{w}_i\Ld\invDtm
    w_i^*=0.
  \end{align*}}
\end{proof}
\begin{theorem}[Darboux transformation and method of variation of constant for 2dTLSCS (\ref{eqns:2dTLSCS-1})]
  \label{thm:DT-src-adjsrc-new}
  Let $f$ and $g$ be two linear independent solutions to (\ref{eqns:Lax-2TodaEqScs}). Suppose $a(y)$ is
  arbitrary functions of time $y$. Let $h:=f+a(y)g$,
  \begin{subequations}
    \label{eqns:DT-new-src-adjsrc}
    \begin{align}
      \tld{w}_{N+1}&=c(y)\mc{D}(f),\label{eqn:DT-new-src}\\
      \tld{w}_{N+1}^*&=\frac{d(y)}{\shft{h}{1}},\label{eqn:DT-new-adjsrc}
    \end{align}
  \end{subequations}
  then (\ref{eqns:DT-toda}), (\ref{eqns:DT-src-adjsrc}) and (\ref{eqns:DT-new-src-adjsrc}) give a new solution
  for (\ref{eqns:2dTLSCS-1}) and (\ref{eqns:Lax-2TodaEqScs}) with $N+1$ self-consistent sources , where
  $c(y)$, $d(y)$ satisfy $c(y)d(y)=\p_y\log a(y)$.
\end{theorem}
\begin{proof}
  It is easy to see $\tld{w}_{N+1}$ satisfies (\ref{eqn:Lax-2TodaEqScs-src}). To prove that $\tld{w}_{N+1}^*$
  satisfies (\ref{eqn:Lax-2TodaEqScs-adjsrc}), we have
  \begin{align*}
    \tld{w}_{N+1,x}^*&=-\frac{d(y)}{\shft{h}{1}{}^2}\shft{h}{1}_x =
    -\frac{d(y)}{\shft{h}{1}{}^2}(\shft{h}{2}+\shft{u}{1}\shft{h}{1}),\\
    -\tld{w}_{N+1}^*\shft{{}}{-1}-\tld{u}\tld{w}_{N+1}^*&=
    -\frac{d(y)}{h}-(\shft{u}{1}-\frac{\shft{h}{1}}{h}
    +\frac{\shft{h}{2}}{\shft{h}{1}})\frac{d(y)}{\shft{h}{1}}=
    -\frac{d(y)\shft{h}{2}}{\shft{h}{1}{}^2}-\shft{u}{1}\frac{d(y)}{\shft{h}{1}}.
  \end{align*}
  Based on proposition \ref{thm:DT-src-adjsrc}, we want to show that extra terms come out from
  \begin{displaymath}
    \mc{D}_y+\mc{D}\left(C+\sum_{i=1}^Nw_i\invDtm w_i^*\right)-
    \left(\tld{C}+\sum_{i=1}^{N+1}\tld{w_i}\invDtm\tld{w}_i^*\right)\mc{D}
  \end{displaymath}
  can be cancled out. It because
  \begin{align*}
    &-\frac{a_y\shft{g}{1}}{h}+\frac{\shft{h}{1}a_y
      g}{h^2}-\tld{w}_{N+1}\invDtm\tld{w}_{N+1}^*(\Ld-\shft{h}{1}/h)\\
    =&\frac{a_y}{h^2}\cas(g,h)-c(y)d(y)\frac{\cas(h,f)}{h}
    \invDtm(\shft{h}{1}{}^{-1}\Ld-\frac{1}{h})\\
    =&\frac{a_y}{h^2}\cas(g,f)-c(y)d(y)a(y)\frac{\cas(g,f)}{h^2}=0.
  \end{align*}
\end{proof}

\subsection{$m$-time repeated non-auto-B\"acklund DTs}
\label{subsec:m-DTs}

\begin{theorem}
  Let $f_j$ and $g_j$~($j=1,2,\ldots,m$) be $m$ pairs of independent solutions to
  (\ref{eqns:Lax-2TodaEqScs}). Suppose $a_j(y)$ are arbitrary functions of time. Let
  \begin{displaymath}
    h_j:=f_j+a_j(y)g_j.
  \end{displaymath}
  Then after $m$-time repetition of Theorem \ref{thm:DT-src-adjsrc-new}, we find a solution for
  (\ref{eqns:Lax-2TodaEqScs}) with $N+m$ self-consistent sources, which is
  \begin{subequations}
    \label{eqns:n-DT}
    \begin{align}
      u[m]&=\shft{u}{m}
      +\frac{\shft{\tcas}{1}(h_1,\cdots,h_m)}{\shft{\cas}{1}(h_1,\cdots,h_m)}-
      \frac{\tcas(h_1,\cdots,h_m)}{\cas(h_1,\cdots,h_m)},\label{eqn:n-DT-u}\\
      v[m]&=v\frac{\shft{\cas}{1}(h_1,\cdots,h_m)\shft{\cas}{-1}(h_1,\cdots,h_m)}
      {\cas^2(h_1,\cdots,h_m)},\label{eqn:n-DT-v}\\
      w_i[m]&=\frac{\cas(h_1,\cdots,h_m,w_i)}{\cas(h_1,\cdots,h_m)},\label{eqn:n-DT-src}\quad
      i=1,\ldots,N,\\
      w_i^*[m]&=(-1)^m\frac{\bcas(h_1,\cdots,h_m,w_i^*)}{\shft{\cas}{1}(h_1,\cdots,h_m)},
      \label{eqn:n-DT-adjsrc}\\
      w_{N+j}[m]&=c_j(y)f_j[m]=c_j(y)\frac{\cas(h_1,\cdots,h_m,f_j)}{\cas(h_1,\cdots,h_m)},\quad
      j=1,\ldots,m
      \label{eqn:n-DT-src-new}\\
      w_{N+j}^*[m]&=(-1)^{m-j}d_j(y)\frac{\shft{\cas}{1}(h_1,\cdots,\hat{h}_j,\cdots,h_m)}
      {\shft{\cas}{1}(h_1\cdots,h_m)},\label{eqn:n-DT-adjsrc-new}
    \end{align}
  \end{subequations}
  where
  \begin{displaymath}
    \tcas(h_1,\cdots,h_m)=
    \begin{vmatrix}
      h_1 & \cdots & h_m\\
      \vdots & & \vdots\\
      \shft{h_1}{m-2} & \cdots & \shft{h_m}{m-2}\\
      \shft{h_1}{m} & \cdots &\shft{h_m}{m}
    \end{vmatrix},
    \bcas(h_1,\cdots,h_m,f)=
      \begin{vmatrix}
        \Sm(h_1f) & \cdots & \Sm(h_mf)\\
        \shft{h_1}{1} &\cdots &\shft{h_m}{1}\\
        \vdots & & \vdots\\
        \shft{h_1}{m-1} & \cdots & \shft{h_m}{m-1}
      \end{vmatrix}
  \end{displaymath}
  and $c_j(y)d_j(y)=\p_y\log a_j(y)$.
\end{theorem}
\begin{proof}
  Since each time Darboux transformation has the form $\mc{D}=\Ld+\sg$, after $m$-time repetition,
  corresponding operator has the form
  \begin{displaymath}
    \mc{D}(m)=\Ld^m+\sg_{m-1}\Ld^{m-1}+\cdots+\sg_0.
  \end{displaymath}
  There are $m$ indetermined coefficients $\sg_i$, $i=0,\ldots,m-1$.
  From (\ref{eqn:annihilate}) we know
  \begin{displaymath}
    \mc{D}(m)h_i=0,\quad i=1,2,\ldots,m.
  \end{displaymath}
  So the indertermined coefficients satisfies
  \begin{displaymath}
    \left[
      \begin{matrix}
        h_1 & \shft{h_1}{1} & \cdots & \shft{h_1}{m-1}\\
        h_2 & \shft{h_2}{1} & \cdots & \shft{h_2}{m-1}\\
        \vdots & \vdots & \vdots & \vdots\\
        h_m & \shft{h_m}{1} & \cdots & \shft{h_m}{m-1}
      \end{matrix}
    \right]\left[
      \begin{matrix}
        \sg_0\\
        \sg_1\\
        \vdots\\
        \sg_{m-1}
      \end{matrix}\right]
    =-\left[
      \begin{matrix}
        \shft{h_1}{m}\\
        \shft{h_2}{m}\\
        \vdots\\
        \shft{h_m}{m}
      \end{matrix}
    \right].
  \end{displaymath}
  By Cramer rule, it is easy to see
  \begin{displaymath}
    \sg_0=(-1)^m\frac{\shft{\cas}{1}(h_1,\cdots,h_m)}{\cas(h_1,\cdots,h_m)},\quad
    \sg_{m-1}=-\frac{\tcas(h_1,\cdots,h_m)}{\cas(h_1,\cdots,h_m)}.
  \end{displaymath}
  Note that DT of $u$, $v$ are the same for 2dTL and 2dTLSCS. So we may omit source term temporary. That is,
  assuming that $\mc{D}(m)$ transforms $\psi$ to $\tld{\psi}=\psi[m]$, and satisfying
  $\tld{\psi}_x=\tld{B}\tld{\psi}$ and $\tld{\psi}_y=\tld{C}\tld{\psi}$, then we have {\allowdisplaybreaks
    \begin{subequations}
      \label{eqns:n-DT-condition}
      \begin{align}
        &\mc{D}(m)_x+\mc{D}(m)B-\tld{B}\mc{D}(m)=0,\label{eqn:n-DT-condition-x}\\
        &\mc{D}(m)_y+\mc{D}(m)C-\tld{C}\mc{D}(m)=0.\label{eqn:n-DT-condition-y}
      \end{align}
    \end{subequations}} Comparing the coefficient of $\Ld^m$ in (\ref{eqn:n-DT-condition-x}), we have
  (\ref{eqn:n-DT-u}). Comparing the coefficient of $\Ld^{-1}$ in (\ref{eqn:n-DT-condition-y}), we have
  (\ref{eqn:n-DT-v}). For arbitrary eigenfunction $w$, its DT $\tld{w}=\mc{D}(m)(w)$ can be expressed in a
  compact form (\ref{eqn:n-DT-src}) according to the Laplace expansion formula. For (\ref{eqn:n-DT-adjsrc}),
  we need induction. Suppose for any adjoint eigenfunction $w^*$, the $m$-time DT formula is correct, then by
  (\ref{eqn:DT-adj-src}), the $m+1$-th DT is
  \begin{align*}
    &w^*[m+1]=-\frac{\Sm(h_{m+1}[m]w^*[m])}{\shft{h_{m+1}[m]}{1}}\\
    &=\frac{\Sm\left[\Dt\left(\frac{h_{m+1}[m-1]}{h_m[m-1]}\Sm(h_m[m-1]w^*[m-1])\right)-
        \shft{h_{m+1}[m-1]}{1} \shft{w^*[m-1]}{1}\right]}{\shft{h_{m+1}[m]}{1}}\\
    &=\frac{\shft{h_{m+1}[m-1]}{1}\Sm(h_m[m-1]w^*[m-1])}{\shft{h_{m+1}[m]}{1}\shft{h_m[m-1]}{1}}
    -\frac{\Sm(h_{m+1}[m-1]w^*[m-1])}{\shft{h_{m+1}[m]}{1}}.
  \end{align*}
  By assumption
  \begin{align*}
    &\frac{\Sm(h_m[m-1]w^*[m-1])}{\shft{h_m[m-1]}{1}}=(-1)^{m+1}\frac{\bcas(h_1,\cdots,h_m,w^*)}
    {\shft{\cas}{1}(h_1,\cdots,h_m)},\\
    &\frac{\Sm(h_{m+1}[m-1]w^*[m-1])}{\shft{h_{m+1}[m-1]}{1}}=(-1)^{m+1}
    \frac{\bcas(h_1,\cdots,h_{m-1},h_{m+1},w^*)}{\shft{\cas}{1}(h_1,\cdots,h_{m-1},h_{m+1})},
  \end{align*}
  we know
  {\allowdisplaybreaks
    \begin{align*}
      &w^*[m+1]=\frac{(-1)^{m+1}}
      {\shft{\cas}{1}(h_1,\cdots,h_{m-1})\shft{\cas}{1}(h_1,\cdots,h_{m+1})}\\
      &\times
      \bigg[\bcas(h_1,\cdots,h_m,w^*)\shft{\cas}{1}(h_1,\cdots,h_{m-1},h_{m+1})\\
      &\quad-\bcas(h_1,\cdots,h_{m-1},h_{m+1},w^*)\shft{\cas}{1}(h_1,\cdots,h_m)\bigg]\\
      =&\frac{(-1)^{m+1}}
      {\shft{\cas}{1}(h_1,\cdots,h_{m-1})\shft{\cas}{1}(h_1,\cdots,h_{m+1})}\\
      &\times
      \begin{vmatrix}
        \Sm(h_1w^*) &\cdots & \Sm(h_{m-1}w^*) & \Sm(h_mw^*) &
        \Sm(h_{m+1}w^*) & 0 &\cdots &0\\
        \shft{h_1}{1} &\cdots &\shft{h_{m-1}}{1} & \shft{h_m}{1} &
        \shft{h_{m+1}}{1} & 0 &\cdots & 0\\
        \vdots & & \vdots &\vdots &\vdots &\vdots & &\vdots\\
        \shft{h_1}{m-1} & \cdots &\shft{h_{m-1}}{m-1} &\shft{h_m}{m-1}
        &\shft{h_{m+1}}{m-1} & 0 &\cdots &0\\
        \shft{h_1}{1} &\cdots &\shft{h_{m-1}}{1} &\shft{h_m}{1}
        &\shft{h_{m+1}}{1} &\shft{h_1}{1} &\cdots &\shft{h_{m-1}}{1}\\
        \vdots & &\vdots &\vdots &\vdots &\vdots & &\vdots\\
        \shft{h_1}{m} &\cdots &\shft{h_{m-1}}{m} &\shft{h_m}{m}
        &\shft{h_{m+1}}{m} &\shft{h_1}{m} &\cdots &\shft{h_{m-1}}{m}
      \end{vmatrix}\\
      =&(-1)^{m+1}\frac{\bcas(h_1,\cdots,h_{m+1},w^*)}{\shft{\cas}{1}(h_1,\cdots,h_{m+1})}.
    \end{align*}}

  Next we want to prove the $m$-th repetition formula for new source terms. Firstly, it is easy to see
  $w_{N+j}[m]=c_j(y)f_j[m]$ can be derived from (\ref{eqn:n-DT-src-new}). For $w_{N+j}^*[m]$, we use
  induction. Suppose when $j\le m$ the formula for $w_{N+j}^*[m]$ is given by (\ref{eqn:n-DT-adjsrc-new}), then
  \begin{displaymath}
    w_{N+j}^*[m+1]=-\frac{\Sm(h_{m+1}[m]w_{N+j}^*[m])}{\shft{h_{m+1}[m]}{1}}.
  \end{displaymath}
  Because $h_{m+1}[m]$ is obtained by $m$-time repetition of DT by using $h_1\cdots,h_m$ sequentially, which is
  equivalent to $m$-time repetition of DT by successively using $h_1,\cdots,h_{j-1},h_{j+1},\cdots,h_{m}$ and
  at last $h_j$,
  \begin{displaymath}
    h_{m+1}[m]=\shft{h_{m+1}[m-1]}{1}-\frac{\shft{h_j[m-1]}{1}}{h_j[m-1]}h_{m+1}[m-1].
  \end{displaymath}
  Note that (\ref{eqn:n-DT-adjsrc-new}), we have $w_{N+j}^*[m]=\frac{d_j(y)}{\shft{h_j[m-1]}{1}}$, so
  \begin{align*}
    w_{N+j}^*[m+1]
    &=-\frac{d_j(y)}{\shft{h_{m+1}[m]}{1}}\Sm\Dt\left(\frac{h_{m+1}[m-1]}{h_j[m-1]}\right)\\
    &=(-1)^{m+1-j}\frac{d_j(y)\shft{\cas}{1}(h_1,\cdots,\hat{h}_j,\cdots,h_{m+1})}
    {\shft{\cas}{1}(h_1,\cdots,h_{m+1})}.
  \end{align*}
  When $j=m+1$,
  \begin{displaymath}
    w_{N+m+1}^*[m+1]=\frac{d_{m+1}(y)}{\shft{h_{m+1}[m]}{1}}
    =\frac{d_{m+1}(y)\shft{\cas}{1}(h_1,\cdots,h_m)}{\shft{\cas}{1}(h_1,\cdots,h_{m+1})}.
  \end{displaymath}
\end{proof}

\section{Solutions for 2dTLSCS}
\label{sec:solutions}

Let us start from trivial solution $q=1$, $v=1$, $u=0$, $N=0$ for 2dTLSCS (\ref{eqns:Lax-2TodaEqScs}). The Lax
pair reads
\begin{subequations}
  \label{eqns:trivial-lax}
  \begin{align}
    \psi_x&=\shft{\psi}{1},\\
    \psi_y&=\shft{\psi}{-1}.
  \end{align}
\end{subequations}

\subsection{Solitons}
\label{subsec:2toda-soliton}

Equations (\ref{eqns:trivial-lax}) have two linearly independent solutions
\begin{displaymath}
  f(n,x,y)=\exp(n\om+z x +z^{-1}y),\quad
  g(n,x,y)=\exp(-n\om+ z^{-1}x + zy),
\end{displaymath}
where $z=e^\om$. Let $a(y)=e^{\af(y)}$, then
\begin{displaymath}
  h=f+a(y) g=2\exp\Om\cdot\cosh Z,
\end{displaymath}
here
\begin{displaymath}
  \Om=\cosh\om\cdot x + \cosh\om\cdot y + \af/2, \quad
  Z=n\om+\sinh\om\cdot x -\sinh\om\cdot y-\af/2.
\end{displaymath}
By (\ref{eqns:n-DT}), taking $m=1$ we have the following 1-soliton solution for (\ref{eqns:2dTLSCS-1})
\begin{align*}
  u[1]&=\frac{\cosh(Z+2\om)}{\cosh(Z+\om)}-\frac{\cosh(Z+\om)}{\cosh
    Z},\\
  v[1]&=\frac{\cosh(Z+\om)\cosh(Z-\om)}{\cosh^2 Z},\\
  w[1]&=c(y)\frac{\sinh\om\cdot e^\Om}{\cosh Z},\\
  w^*[1]&=\frac{d(y)e^{-\Om}}{2\cosh(Z+\om)},
\end{align*}
where $c(y)d(y)=\dot\af$.

If take two pairs of independent solutions, with respect to $z_j=e^{\om_j}$ ($j=1,2$), i.e.
\begin{displaymath}
  f_j=\exp(n\om_j+z_jx+z_j^{-1}y),\quad
  g_j=\exp(-n\om_j+z_j^{-1}x+z_jy) \quad j=1,2.
\end{displaymath}
Let $a_j(y)=e^{\af_j(y)}$, then
\begin{displaymath}
  h_j=f_j+a_jg_j=2\exp\Om_j\cdot\cosh Z_j,
\end{displaymath}
where
\begin{displaymath}
  \Om_j=\cosh\om_j\cdot x+\cosh\om_j\cdot y+\af_j/2,\quad
  Z_j=n\om_j+\sinh\om_j\cdot x-\sinh\om_j\cdot y-\af_j/2.
\end{displaymath}
To simplify the notion, for $k\in\mb{Z}$, define
\begin{align*}
  H_k=&\left|
  \begin{matrix}
    \cosh Z_1 & \cosh Z_2\\
    \cosh (Z_1+k\om_1) & \cosh (Z_2+k\om_2)
  \end{matrix}
  \right|\\
  =&\sinh\frac{k(\om_1-\om_2)}{2}\sinh\left(Z_1+Z_2+\frac{k}{2}(\om_1+\om_2)\right)\\
  &+\sinh\frac{k(\om_1+\om_2)}{2}\sinh\left(Z_1-Z_2+\frac{k}{2}(\om_1-\om_2)\right).
\end{align*}
Then 2-soliton solution for (\ref{eqns:2dTLSCS-1}) is
\begin{align*}
  u[2]&=\frac{\shft{H_2}{1}}{\shft{H_1}{1}}-\frac{H_2}{H_1},\\
  v[2]&=\frac{\shft{H_1}{1}\shft{H_1}{-1}}{H_1^2},\\
  w_1[2]&=c_1(y)a_1(y)\frac{2\sinh\om_1(\cosh\om_1-\cosh\om_2)\exp\Om_1\cosh(Z_2+\om_2)}{H_1},\\
  w_2[2]&=c_2(y)a_2(y)\frac{2\sinh\om_2(\cosh\om_1-\cosh\om_2)\exp\Om_2\cosh(Z_1+\om_1)}{H_1},\\
  w^*_1[2]&=-\frac{d_1(y)\exp(-\Om_1)\cosh(Z_2+\om_2)}{\shft{H_1}{1}},\\
  w^*_2[2]&=\frac{d_2(y)\exp(-\Om_2)\cosh(Z_1+\om_1)}{\shft{H_1}{1}},
\end{align*}
where $c_j(y)d_j(y)=\dot\af_j$.

\subsection{Rational solution}
In equation (\ref{eqns:trivial-lax}), noticing that $\p_z^k\psi$ is another solution. Since
$g(n,x,y)=z^n\exp(zx+z^{-1}y)$ and $f_k(n,x,y):=\p_z^kg$~($k\ge 1$) are all independent solutions for
(\ref{eqns:trivial-lax}). Let $\xi:=zx+z^{-1}y$, then
\begin{align*}
  f_1(n,x,y)&=\p_zg=z^{n-1}e^\xi(n+z\xi_z),\\
  f_2(n,x,y)&=\p_z^2g=z^{n-2}e^\xi(n^2+2nz\xi_z-n+z^2\xi_z^2+z^2\xi_{zz}),\\
  f_3(n,x,y)&=\cdots
\end{align*}
Let $h_k=f_k+a(y)g$. Take $k=1$, $m=1$ in (\ref{eqns:n-DT}), one yields
\begin{align*}
  u[1]&=-\frac{z}{(\eta+za+1/2)^2-1/4},\\
  v[1]&=1-\frac{1}{(\eta+za)^2},\\
  w[1]&=c(y)\frac{z^{n+1}e^\xi a}{\eta+za},\\
  w^*[1]&=d(y)\frac{z^{-n}e^{-\xi}}{\eta+za+1},
\end{align*}
where $\eta=n+z\xi_z$,$c(y)d(y)=\frac{\dr}{\dr y}\log a$. This is a rational solution for
(\ref{eqns:2dTLSCS-1}).

If take $k=2,m=1$, we find another rational solution for (\ref{eqns:2dTLSCS-1})
\begin{align*}
  u[1]&:=z\left(\frac{\shft{\eta}{2}+z^2a}{\shft{\eta}{1}+z^2a}-
    \frac{\shft{\eta}{1}+z^2a}{\eta+z^2a}\right),\\
  v[1]&:=\frac{(\shft{\eta}{1}+z^2a)(\shft{\eta}{-1}+z^2a)}{(\eta+z^2a)^2},\\
  w[1]&:=2c(y)az^{n+1}e^\xi\frac{n+z\xi_z}{\eta+z^2a},\\
  w^*[1]&:=d(y)\frac{z^{-n+1}e^{-\xi}}{\shft{\eta}{1}+z^2a},
\end{align*}
where $\eta=n^2+2nz\xi_z-n+z^2\xi_z^2+z^2\xi_{zz}$, $c(y)d(y)=\frac{\dr}{\dr y}\log a$.

\subsection{Other solutions}
Let
\begin{displaymath}
  f=z^ne^{zx+z^{-1}y}:=z^ne^{F(x,y,z)},\quad
  g=z^{-n}e^{z^{-1}x+zy}:=z^{-n}e^{G(x,y,z)},
\end{displaymath}
be pair of solutions to (\ref{eqns:trivial-lax}), then $f_z$ and $g_z$ are another pair of solutions to
(\ref{eqns:trivial-lax}). Let
\begin{displaymath}
  h=f+a(y)g=2\exp\Om\cdot\cosh Z,
\end{displaymath}
where $\Om$ and  $Z$ are defined in subsection \ref{subsec:2toda-soliton}. Then
\begin{displaymath}
  h_z=f_z+a(y)g_z=2\Om_ze^\Om\cosh Z+2e^\Om Z_z\sinh Z.
\end{displaymath}
From (\ref{eqns:n-DT}), taking $m=2$, we construct solutions with singularities. For simplicity, we define
\begin{displaymath}
  \begin{vmatrix}
    h & h_z\\ \shft{h}{k} & \shft{h_z}{k}
  \end{vmatrix}
  =4e^{2\Om}C_k,\quad\text{where} C_k=(Z_z+\frac{k}{2z})\sinh(k\om)+\frac{k}{2z}\sinh(2Z+k\om).
\end{displaymath}
\begin{align*}
  \cas(h,h_z,f)&=-8a(x)e^{3\Om}\frac{\sinh^2\om}{z}\cosh(Z+\om),\\
  \cas(h,h_z,f_z)&=4a(x)e^{3\Om}\left(\shft{D_1}{1}\cosh Z +D_1\cosh(Z+2\om)-D_2\cosh(Z+\om)\right),
\end{align*}
where
\begin{displaymath}
  D_k=-(\frac{n}{z}+\frac{k}{2z}+F_z)(\frac{n}{z}+\frac{k}{2z}-G_z)\sinh(k\om)
  +\frac{k^2}{4z^2}\sinh(k\om)+\frac{k\Om_z}{z}\cosh(k\om),\;k=1,2.
\end{displaymath}
then find the following solution for (\ref{eqns:2dTLSCS-1})
\begin{align*}
  u[2]&=\frac{\shft{C_2}{1}}{\shft{C_1}{1}}-\frac{C_2}{C_1},\quad
  v[2]=\frac{\shft{C_1}{1}\shft{C_1}{-1}}{C_1^2},\\
  w_1[2]&=-2c_1(y)a(y)\frac{\sinh^2\om e^\Om}{z C_1}\cosh(Z+\om),\\
  w_2[2]&=c_2(y)a(y)\frac{e^\Om\left(\shft{D_1}{1}\cosh Z
      +D_1\cosh(Z+2\om)-D_2\cosh(Z+\om)\right)}{C_1},\\
  w_1^*[2]&=-d_1(y)\frac{\Om_z\cosh(Z+\om)+(Z_z+1/z)\sinh(Z+\om)}{2e^\Om
    \shft{C_1}{1}},\\
  w_2^*[2]&=d_2(y)\frac{\cosh(Z+\Om)}{2e^\Om \shft{C_1}{1}}.
\end{align*}
where $c_j(y)d_j(y)=\frac{\dr}{\dr y}\log a(y)$.

\section*{Conclusion}
We present a new multi-component two dimensional Toda lattice hierarchy, which enables us to find the two
dimensional Toda lattice equation with self-consistent sources in different way from
\cite{MR708435,MR910584,MR940618,MR2261273,2-toda-scs-casorati} as well as their Lax representations. Since
two dimensional Toda lattice equation with self-consistent sources can be considered as a two dimensional Toda
lattice equation with non-homogeneous terms, method of variation of constant can be applied to the ordinary
Darboux transformations for 2dTLSCS to construct a non-auto-B\"acklund Darboux transformations. Then it offers
a different way to solve 2dTLSCS in contrast with \cite{MR2261273,2-toda-scs-casorati}.

The 2dTLH offers various types of reductions, for example periodic reductions and reductions to Toda lattice
equation. It is an interesting question does this new mc2dTLH offers similar reductions. We may discuss such
problems elsewhere.

\section*{Acknowledgement}
The author thanks professor Hu Xingbiao for valuable comments. This work was supported by National Basic
Research Program of China (973 Program) (2007CB814800) and National Natural Science Foundation of China (grand
No. 10601028).


\begin{thebibliography}{10}

\bibitem{MR638807}
Etsur{\=o} Date, Michio Jimbo, Masaki Kashiwara, and Tetsuji Miwa.
\newblock Transformation groups for soliton equations. {III}. {O}perator
  approach to the {K}adomtsev-{P}etviashvili equation.
\newblock {\em J. Phys. Soc. Japan}, 50(11):3806--3812, 1981.

\bibitem{MR723457}
Michio Jimbo and Tetsuji Miwa.
\newblock Solitons and infinite-dimensional {L}ie algebras.
\newblock {\em Publ. Res. Inst. Math. Sci.}, 19(3):943--1001, 1983.

\bibitem{MR730247}
Mikio Sato and Yasuko Sato.
\newblock Soliton equations as dynamical systems on infinite-dimensional
  {G}rassmann manifold.
\newblock In {\em Nonlinear partial differential equations in applied science
  (Tokyo, 1982)}, volume~81 of {\em North-Holland Math. Stud.}, pages 259--271.
  North-Holland, Amsterdam, 1983.

\bibitem{MR688946}
Etsur{\=o} Date, Michio Jimbo, Masaki Kashiwara, and Tetsuji Miwa.
\newblock Transformation groups for soliton equations. {E}uclidean {L}ie
  algebras and reduction of the {KP} hierarchy.
\newblock {\em Publ. Res. Inst. Math. Sci.}, 18(3):1077--1110, 1982.

\bibitem{MR2006751}
V.~G. Kac and J.~W. van~de Leur.
\newblock The {$n$}-component {KP} hierarchy and representation theory.
\newblock {\em J. Math. Phys.}, 44(8):3245--3293, 2003.
\newblock Integrability, topological solitons and beyond.

\bibitem{MR1621464}
Johan van~de Leur.
\newblock Schlesinger-{B}\"acklund transformations for the {$N$}-component
  {KP}.
\newblock {\em J. Math. Phys.}, 39(5):2833--2847, 1998.

\bibitem{MR1629543}
H.~Aratyn, E.~Nissimov, and S.~Pacheva.
\newblock A new ``dual'' symmetry structure of the {KP} hierarchy.
\newblock {\em Phys. Lett. A}, 244(4):245--255, 1998.

\bibitem{MR810623}
Kimio Ueno and Kanehisa Takasaki.
\newblock Toda lattice hierarchy.
\newblock In {\em Group representations and systems of differential equations
  (Tokyo, 1982)}, volume~4 of {\em Adv. Stud. Pure Math.}, pages 1--95.
  North-Holland, Amsterdam, 1984.

\bibitem{MR708435}
V.~K. Mel{\cprime}nikov.
\newblock On equations for wave interactions.
\newblock {\em Lett. Math. Phys.}, 7(2):129--136, 1983.

\bibitem{MR910584}
V.~K. Mel{\cprime}nikov.
\newblock A direct method for deriving a multisoliton solution for the problem
  of interaction of waves on the {$x,y$} plane.
\newblock {\em Comm. Math. Phys.}, 112(4):639--652, 1987.

\bibitem{MR940618}
V.~K. Mel{\cprime}nikov.
\newblock Exact solutions of the {K}orteweg-de {V}ries equation with a
  self-consistent source.
\newblock {\em Phys. Lett. A}, 128(9):488--492, 1988.

\bibitem{MR2261273}
Xing-Biao Hu and Hong-Yan Wang.
\newblock Construction of d{KP} and {BKP} equations with self-consistent
  sources.
\newblock {\em Inverse Problems}, 22(5):1903--1920, 2006.

\bibitem{nmcKP}
Xiaojun Liu, Yunbo Zeng, and Runliang Lin.
\newblock A {N}ew {M}ulti-component {KP} {H}ierarchy.
\newblock In submission.

\bibitem{2-toda-scs-casorati}
Hong-Yan Wang, Xing-Biao Hu, and Gegenhasi.
\newblock 2d toda lattice equation with self-consistent sources: Casoratian
  type solutions, bilinear b\"{a}cklund transformation and lax pair.
\newblock {\em J. Comp. Appl. Math.}, 202(1):133--143, 2007.

\bibitem{MR1964513}
L.~A. Dickey.
\newblock {\em Soliton equations and {H}amiltonian systems}, volume~26 of {\em
  Advanced Series in Mathematical Physics}.
\newblock World Scientific Publishing Co. Inc., River Edge, NJ, second edition,
  2003.

\bibitem{bk:MatveevSalle91}
V.~B. Matveev and M.~A. Salle.
\newblock {\em Darboux transformations and solitons}.
\newblock Springer Series in Nonlinear Dynamics. Springer-Verlag, Berlin, 1991.

\end{thebibliography}

\def\cprime{$'$} \def\cprime{$'$} \def\cprime{$'$} \def\cprime{$'$}
  \def\cprime{$'$}

\end{document}